\documentclass{article}
\usepackage{verbatim}

\usepackage{times}
\usepackage[pdftex]{graphicx}
\usepackage{amssymb,amsfonts,amsmath,amsthm}
\usepackage{mathtools}
\usepackage{float}
\usepackage{bbm}
\usepackage{xcolor}
\usepackage{amsfonts}
\usepackage{enumitem}

\newtheorem{theorem}{Theorem}
\newtheorem{proposition}{Proposition}
\newtheorem{lemma}{Lemma}
\newtheorem{ass}{Assumption}
\newtheorem{definition}{Definition}
\newtheorem{remark}{Remark}

\newcommand\EE {\mathbb E}
\newcommand\FF {\mathbb F}

\newcommand\RR {\mathbb R}

\newcommand\cE {\mathcal E}

\newcommand\cN {\mathcal N}

\def\qed{\hskip6pt\vrule height6pt width5pt depth1pt}

\title{Optimal brokerage contracts in Almgren-Chriss model with multiple clients}
\author{G. Alonso Alvarez, S. Nadtochiy, and K. Webster
\footnote{Address all correspondence to Sergey Nadtochiy, lllinois Institute of Technology, Chicago, IL, USA, snadtochiy@hawk.iit.edu.}
\footnote{The first and second authors acknowledge partial support from the NSF CAREER grant DMS-1651294.}
}
\date{April 2022}

\begin{document}

\maketitle

\begin{abstract}
This paper constructs optimal brokerage contracts for multiple (heterogeneous) clients trading a single asset whose price follows the Almgren-Chriss model. The distinctive features of this work are as follows: (i) the reservation values of the clients are determined endogenously, and (ii) the broker is allowed to not offer a contract to some of the potential clients, thus choosing her portfolio of clients strategically. We find a computationally tractable characterization of the optimal portfolios of clients (up to a digital optimization problem, which can be solved efficiently if the number of potential clients is small) and conduct numerical experiments which illustrate how these portfolios, as well as the equilibrium profits of all market participants, depend on the price impact coefficients.
\end{abstract}

\medskip


\section{Introduction}

This paper investigates the optimal design of brokerage fees in Almgren-Chriss model (see \cite{AlmgrenChriss}). We consider a population of investors (a.k.a. potential clients or agents) who can either trade directly in the market (and be subject to trading costs due to their price impact) or to trade via a broker (i.e., to become broker's clients) who charges a contingent fee for this service. The main goal of our investigation is a tractable characterization of the optimal brokerage fees and the optimal choice of a portfolio of clients. 
This question is formulated as an optimal contract problem with multiple agents, where the broker plays the role of a principal who designs the fees. 

The problem considered herein formally fits within the optimal contract theory, which is concerned with the design of compensation (or incentive) schemes, referred to as contracts. In the classical example of an optimal contract problem (see, among others, \cite{HolmstromMilgrom}, \cite{CvitanicZhang}), a principal hires an agent to work on a project in exchange for a payment (contract). The payment depends on the information available to the principal which may be affected by the agent's action. The agent chooses his action to maximize his objective, which depends on the payment promised by the principal and on the action itself (e.g., the agent may not like to work very hard). The principal aims to choose the contract so that it maximizes her objective, which also depends on the payment to the agent and on the agent's action. This leads to a pair of nested optimization problems, also known as the Stackelberg game.

In the present setting, the broker observes the trading strategy of her client precisely, which leads to a so-called first best optimal contact problem between the broker and her individual client. Such optimal contract problems are known to have simple solutions (especially in the case of risk-neutral preferences, as herein), and this is confirmed by Proposition \ref{prop:multiagents.optContracts} which provides the optimal brokerage fees in the present setting. However, our setting implies the following additional challenges. First, the model considered herein includes multiple agents whose objectives are coupled via their price impact. Thus, a collective response of the agents to a contract chosen by the principal is given by a collection of strategies that form a Nash equilibrium among the agents. The principal, then, chooses a contract so as to maximize her objective that depends on the associated equilibrium strategies of the agents. The optimal contract problems with multiple agents are considered, for example, in \cite{ElliePossamai}, \cite{CasgrainJaimungal}, \cite{HuangJaimungal}. It is worth mentioning that we consider heterogenous agents, as one of our goals is to study how the characteristics of the agents (i.e., their price impact coefficients) affect the optimal choice of the portfolio of clients.
Another important feature that makes the present problem non-standard is the fact that the reservation value of each agent (which represents the minimum objective value that the agent must be able to attain in order to accept a proposed contract) is determined endogenously. Indeed, we naturally assume that the reservation value of each agent equals his maximum objective value in case he decides to trade directly in the market. The latter value depends on the equilibrium strategies of other agents, which in turn depend on the contract chosen by the broker. The third distinctive feature of the present work is that, unlike the classical optimal contract problems, the broker is not constrained to offer a contract to every agent and choses her portfolio of clients strategically. In particular, one of our main questions is to determine the optimal portfolio of clients for the broker.

To the best of our knowledge, to date there exist no results on the optimal brokerage fees in the presence of price impact and multiple agents. The recent paper \cite{Frei} studies a related problem in which a single agent hires a financial intermediary to trade a risky asset on his behalf, and pays a fee (chosen by the agent) for this service. This is also related to earlier literature on delegated portfolio management: see, e.g., \cite{Starks}, \cite{Stoughton}, \cite{OuYang}, \cite{CadenillasCvitanic}, \cite{BasakPavlova}, \cite{CvitanicXing}, \cite{NZ}. Despite obvious similarities, the important conceptual differences between the latter works and the present one are that, herein, (i) the fee is designed by the broker, (ii) multiple agents are present, and (iii) ex ante the trading strategy is determined by the client as opposed to the broker, who nevertheless does observe the strategy.

\smallskip

The rest of the paper is organized as follows. Section \ref{se:setup} introduces the model and the main objectives. Section \ref{se:optContract} constructs optimal contracts (Proposition \ref{prop:multiagents.optContracts}) given (arbitrary) reservation values of the agents and broker's (arbitrary) choice of clients. Section \ref{se:equilibrium} describes the unique equilibrium among those agents who are not offered a contract (i.e., among independent agents). Section \ref{se:optimalOF} defines the reservation values of the agents endogenously (Definition \ref{def:endog.R}) and shows how to compute the maximum objective value of the broker given an arbitrary portfolio of clients (Theorem \ref{thm:multiagents.main}). The latter result allows one to find an optimal portfolio of clients for the broker by solving a digital optimization problem. Section \ref{se:numerics} considers several numerical experiments, where the aforementioned digital optimization problem is solved by exhaustive search and the broker's optimal portfolio of clients, as well as the profits of the broker and of the agents, are analyzed as functions of the price impact parameters.

\section{The setup}
\label{se:setup}


We consider $N$ agents, each of whom can trade a single risky asset\footnote{A riskless asset is implicitly available but yields zero return.} that follows the Almgren-Chriss model over the time interval $[0,T]$. In addition to the agents, we assume the presence of a single broker. Each agent makes a decision (once, before the trading starts) on whether he trades the asset directly or via the broker, and these decisions are represented by the vector $\theta \in \{0,1\}^N$: $\theta_i=1$ if and only if the $i$-th agent trades via the broker. For convenience, we also denote by $\cN(\theta) = \{n_1,\cdots, n_r\} \subset \{1,\cdots,N\}$, $0<r<N$, the indices of the agents that trade via broker. We refer to the agents who trade via the broker as clients and to those who trade directly in the market as independent. The trading activity of an independent agent affects the price of the asset via the impact coefficients of this agent. The trading of a client is done via the broker and hence affects the price of the asset via the impact coefficients of the broker. Of course, the broker may be able to offer lower price impact to her clients, but she also charges each of them a fee for this service. Even though we introduced $\theta$ as a vector of agents' decisions, it is important to realize that these decisions are ultimately controlled by the broker, who decides whether to offer a contract to a particular agent or not (the acceptance of each offered contract is ensured by matching the reservation value of the associated agent). Therefore, in the remainder of the paper, we refer to $\theta$ as the choice of clients made by the broker.

We fix a probability space $(\Omega,\mathcal{F},\mathbb{P})$ and consider a standard Brownian motion $B$ on this space. Let $\FF^B$ be the filtration generated by the Brownian motion $B$. We define the set of admissible controls of a single agent as 
\begin{equation}
    \mathcal{U} : = \{\nu \in L^2([0,T]\times\Omega),\quad \nu\text{ is }\FF^B\text{-adapted} \}.
\end{equation}
The (controlled) inventory of agent $i=1,\ldots,N$, who uses a control $\nu^i$, is given by the process
$$
X^i_t = x^i_0 + \int_0^t \nu^i_s ds,\quad t\in[0,T].
$$

Next, we recall the price process for the traded asset in the Almgren-Chriss model: 
\begin{equation}
\label{states}
P_t =  \mu t + \sigma B_t  + \sum_{i=1}^N (1-\theta_i)\kappa_i \nu_{t}^i +\kappa_0 \sum_{i=1}^N \theta_i \nu_{t}^i  + \sum_{i=1}^N (1-\theta_i)\lambda_i X_t^i +\lambda_0 \sum_{i=1}^N \theta_i X^i_t,
\end{equation}
where $\{\kappa_j\}_{j=1}^N$, $\{\lambda_j\}_{j=1}^N$ are the coefficients of temporary and permanent price impacts of the agents, $\kappa_0$, $\lambda_0$ are the corresponding coefficients of the broker, $\sigma>0$ is the volatility of the asset price, and $\mu\in\RR$ is its drift (i.e., trading signal). For convenience, we assume that $\nu_T^i=0$ for all $i$, so that the temporary impacts of the agents do not affect the terminal price.\footnote{This is needed to simplify the notation in \eqref{eq.sec2.clientObj.def}--\eqref{eq.sec2.indepObj.def}, as the agents in our setting interact through both the permanent and the temporary impacts.}

\medskip

Let $\xi^{n_1},\cdots, \xi^{n_r}$ be the fees that the broker charges to her clients. 
We assume that the fees are of the form $\xi^i = F^i(X^i,P)$, with measurable $F^i: H^{1}([0,T],\RR)\times C([0,T],\mathbb{R})\rightarrow \mathbb{R}$, where $H^{1}$ is the Sobolev space of order one, equipped with the natural norm. A client $i \in \cN(\theta)$ aims to maximize his expected profit:
\begin{equation}\label{eq.sec2.clientObj.def}
    \begin{split}
    J^{i,\theta}(\nu^i,\nu^{-i},\xi^i) &= \EE \left(P_T X^i_T
    -\frac{\lambda_i}{2}(X^i_T)^2-\int_0^T \nu^i_t P_tdt -\xi^i\right)  \\
    &=\EE \left( \int_0^T X^i_t\left[\mu+\sum_{j=1}^N (1-\theta_j)\lambda_j \nu^j_t+\lambda_0\sum_{j=1}^N\theta_j\nu^j_t\right]dt \right.\\
    &\left.-\int_0^T \nu^i_t\left[\sum_{j=1}^N (1-\theta_j)\kappa_j \nu^j_t+\kappa_0\sum_{j=1}^N\theta_j\nu^j_t\right] dt - \frac{\lambda_i}{2}(X^i_T)^2-\xi^i\right),
    \end{split}
\end{equation}
where $\nu^{-i}\in \mathcal{U}^{N-1}$ denotes the trading rate of the rest of agents. 
Similarly, an independent agent $i\notin\cN(\theta)$ maximizes his expected profit:
\begin{equation}\label{eq.sec2.indepObj.def}
\begin{split}
J^{i,\theta}(\nu^i,\nu^{-i},\xi^i) &= \EE \left(P_T X^i_T-\frac{\lambda_i}{2}(X^i_T)^2-\int_0^T \nu^i_t P_tdt \right)  \\
    &=\EE \left( \int_0^T X^i_t\left[\mu+\sum_{j=1}^N (1-\theta_j)\lambda_j \nu^j_t+\lambda_0\sum_{j=1}^N\theta_j\nu^j_t\right]dt \right.\\
    &\left.-\int_0^T \nu^i_t\left[\sum_{j=1}^N (1-\theta_j)\kappa_j \nu^j_t+\kappa_0\sum_{j=1}^N\theta_j\nu^j_t\right]dt-\frac{\lambda_i}{2}(X^i_T)^2\right).
    \end{split}
\end{equation}

\begin{remark}
In the above objectives, the penalty $-\lambda_i (X^i_T)^2/2$ corresponds to the cost of liquidation of the inventory $X^i_T$ by agent $i$, over a long time period after time $T$. Even if the agent plans to liquidate his terminal position $X^i_T$ with the broker, in the resulting extended game the broker will charge the agent more by increasing the deterministic component of $\xi^i$. Due to the optimality of the contract, this additional charge should be equal to the liquidation cost of the agent. Assuming that the agent does not know the inventories, and hence the execution plans, of the other agents, he computes his execution cost assuming that the additional drift in the asset price during his liquidation is only due to his permanent impact (the temporary impact can be ignored if the liquidation period is large). This can be interpreted as the agent being naive, in the sense that he does not take into account that other agents will be liquidating their positions at the same time, hence, he only includes his own impact in the execution cost. Under such an assumption it is well known that the expected execution cost of the agent is equal to $\lambda_i (X^i_T)^2/2$, which is added to the fees that the agent pays to the broker.
\end{remark}

For a given combination of the strategies of other agents, we define the control problem of a client $i\in \cN(\theta)$ and of an independent agent $i\notin \cN(\theta)$, respectively, as: 
\begin{equation}\label{eq.sec2.clientOpt.def}
   V^{i,\theta}:=\sup_{\nu^i\in \mathcal{U}} J^{i,\theta}(\nu^i,\nu^{-i},\xi^i), \quad i\in \cN(\theta),
\end{equation}
\begin{equation}\label{eq.sec2.indepOpt.def}
   V^{i,\theta}:=\sup_{\nu^i\in \mathcal{U}} J^{i,\theta}(\nu^i,\nu^{-i}), \quad i\notin \cN(\theta).
\end{equation}

\begin{definition}\label{def:cE}
Given a choice of clients $\theta\in\{0,1\}^N$, as well as the associated indices $\cN(\theta) = \{n_1,\cdots, n_r\}$ and fees $\xi=(\xi^{n_1},\cdots,\xi^{n_r})$, we define $\cE(\theta,\xi)\subset \mathcal{U}^{N}$ as the set of all agents' strategies that form Nash equilibria in the game defined by \eqref{eq.sec2.clientOpt.def}--\eqref{eq.sec2.indepOpt.def}. Namely, $\bar{\nu} \in \cE(\theta,\xi)$ if and only if the following two conditions hold:
\begin{equation}
     J^{i,\theta}(\bar{\nu}^i,\bar{\nu}^{-i},\xi^i) \geq  J^{i,\theta}(\nu,\bar{\nu}^{-i},\xi^i), \quad \forall \nu \in \mathcal{U},
     \quad \forall i \in \cN(\theta),
\end{equation}
\begin{equation}
     J^{i,\theta}(\bar{\nu}^i,\bar{\nu}^{-i}) \geq  J^{i,\theta}(\nu,\bar{\nu}^{-i}), \quad \forall \nu \in \mathcal{U},
     \quad \forall i \notin \cN(\theta).
\end{equation}
\end{definition}

\medskip

The objective of the broker is given by the sum of expected fees in the best equilibrium attainable with these fees:
\begin{equation}\label{eq.JP.def}
    J^\theta_P(\xi) := \sup_{\nu\in \cE(\theta,\xi)}\EE\sum_{j\in\cN(\theta)} \xi^{j}.
\end{equation}
In the above, we make the standard assumption that, given a set of admissible contracts, the agents will choose an equilibrium that is best for the principal among all attainable equilibria.
To ensure that $\cE(\theta,\xi)\neq\emptyset$ and that the agents' reservation values are met, we introduce the set of admissible fees of the broker:
\begin{align}
&\Sigma(\theta) := \{ (\xi^{j})_{j\in\cN(\theta)} : \cE(\theta,\xi)\neq \emptyset,
\text{ and }  J^{j}(\bar{\nu},\xi^{j})\geq R^{j,\theta}, \quad  \forall j \in \cN(\theta),\, \forall \bar{\nu}\in \cE(\theta,\xi) \},
\label{eq.sec2.Sigma.def}
\end{align}
where $R^{j,\theta}$ is the reservation value of agent $j\in \cN(\theta)$.
Thus, we obtain the following ``local" maximization problem for the broker, given a choice of clients $\theta$:
\begin{equation}\label{eq.multiagents.broker.Opt.givenTheta}
    V^\theta_P := \sup_{\xi \in \Sigma(\theta)} J^\theta_P(\xi)=\sup_{\xi \in \Sigma(\theta)} \sup_{\nu\in \cE(\theta,\xi)}\EE \sum_{j\in\cN(\theta)} \xi^{j}.
\end{equation}
The optimal contract $\xi^*$ that attains the above supremum is constructed in Section \ref{se:optContract}.
Note that, in the present setting we do not assume that $R^{j,\theta}$ is given endogenously, because there is in fact a very natural endogenous definition of the reservation value of each agent. However, as the above optimal contract problem is of first-best type, with risk-neutral preferences, the value of $R^{j,\theta}$ is not important for the form of the optimal contract we design in Section \ref{se:optContract}. Therefore, in order to ease the notation, we postpone the definition of $R^{j,\theta}$ to Section \ref{se:optimalOF}.

Note that the ``global" optimization problem of the broker is to find an optimal $\theta^*\in \{0,1\}^N$ and respective optimal fees $\xi^{*}$, which amounts to solving
\begin{equation}\label{eq.multiagents.broker.Opt.global}
    V_P := \sup_{\theta \in\{0,1\}^N} V^\theta_P.
\end{equation}
The above is a discrete optimization problem. We do not provide a complete solution to this problem herein, assuming instead that it can be solved by an exhaustive search in case of a reasonably small $N$ or of a smaller subset of admissible $\{\theta\}$ (see Section \ref{se:numerics}). However, even to perform such an exhaustive search, one needs to have a numerically tractable representation of the value of $V^\theta_P$ for each $\theta$. The latter representation is the main subject of Sections \ref{se:equilibrium},\ref{se:optimalOF}.

\section{Optimal contract for a given $\theta$}
\label{se:optContract} 

The main result of this section is the following proposition which describes an optimal collection of fees offered by the broker, given a choice of clients $\theta$ and agents' reservation values $\{R^{i,\theta}\}$.

\begin{proposition}\label{prop:multiagents.optContracts}
For any $\theta\in\{0,1\}^N$ and $\{R^{i,\theta}\}$, the fees 
\begin{equation}
    \xi^{i,*}:=-R^{i,\theta}+P_TX^i_T - \int_0^T \nu^i_tP_tdt-\frac{\lambda_i}{2}(X^i_T)^2,\quad i \in\cN(\theta),
    \label{fee_optimal}
\end{equation}
are optimal for the problem \eqref{eq.multiagents.broker.Opt.givenTheta}. 
\end{proposition}
\begin{proof}
It is easy to see that, with the fees \eqref{fee_optimal}, the objective of any agent $i\in \cN$, given by \eqref{eq.sec2.clientObj.def}, is constant and is equal to his reservation $R^{i,\theta}$. In particular, every agent $i\in \cN$ is indifferent in which action to choose. Using this observation, let us show that, for any set of admissible fees $(\xi^{j})_{j\in\cN}\in \Sigma(\theta)$, we have the inclusion: 
    \begin{equation}
        \cE(\theta,\xi) \subset \cE(\theta,\xi^*),
        \label{inclusion1}
    \end{equation}
where $\xi^*$ is given by \eqref{fee_optimal}.
Indeed, for any $(\tilde{\nu}^1,\cdots,\tilde{\nu}^N)\in \cE(\theta,\xi)$ and $i\in \cN$, we have: 
\begin{align*}
&J^{i,\theta}(\tilde{\nu}^i,\tilde{\nu}^{-i},\xi^{i,*}) = R^{i,\theta} = \sup_{\nu^i\in \mathcal{U}} J^{i,\theta}(\nu^i,\tilde{\nu}^{-i},\xi^{i,*}),\\
&J^{i,\theta}(\tilde{\nu}^i,\tilde{\nu}^{-i}) = \sup_{\nu^i\in \mathcal{U}} J^{i,\theta}(\nu^i,\tilde{\nu}^{-i}),
\end{align*}
which implies $(\tilde{\nu}^1,\cdots,\tilde{\nu}^N)\in \cE(\theta,\xi^*)$.
Using \eqref{inclusion1} and the admissibility constraint in \eqref{eq.sec2.Sigma.def}, we deduce: 
\begin{equation}\label{optimality}
\begin{split}
          V^\theta_P &=\sup_{\xi \in \Sigma(\theta)} \sup_{\nu\in \cE(\theta,\xi)}\EE \sum_{j\in\cN(\theta)} \xi^{j}
          =\sup_{\xi \in \Sigma(\theta)} \sup_{\nu\in \cE(\theta,\xi)} \left[
          \EE \sum_{j\in\cN} \left(X_T^{j}(P_T-\frac{\lambda_j}{2}X^j_T) -\int_0^T \nu^{j}_t P^\theta_t dt  \right)\right.\\
          &\left.- \EE \sum_{j\in\cN} \left( X_T^{j}(P_T-\frac{\lambda_j}{2}X^j_T) - \int_0^T \nu^{j}_t P^\theta_t dt - \xi^j \right)
          \right]\\
          &\leq \sup_{\xi \in \Sigma(\theta)} \sup_{\nu\in \cE(\theta,\xi)}
          \EE \sum_{j\in\cN} \left(X_T^{j}(P_T-\frac{\lambda_j}{2}X^j_T) -\int_0^T \nu^{j}_t P^\theta_t dt  \right)
          - \sum_{j\in\cN} R^{j,\theta}\\
          &\leq \sup_{\nu\in \cE(\theta,\xi^*)} \EE\left( \sum_{j\in\cN} X_T^{j}(P_T-\frac{\lambda_j}{2}X^j_T) -\int_0^T \nu^{j}_t P^\theta_t dt  \right)-\sum_{j\in\cN} R^{j,\theta}
          = \sup_{\nu\in \cE(\theta,\xi^*)} \EE\left( \sum_{j\in\cN} \xi^{j,*}  \right) = J^\theta_P(\xi^{*}).
\end{split}
\end{equation}
\qed
\end{proof}

\section{Equilibrium strategies of independent agents}
\label{se:equilibrium}

Proposition \ref{prop:multiagents.optContracts} shows that, for any given $\theta$, there exists a trivial choice of optimal contracts. This provides a solution to the broker's local problem \eqref{eq.multiagents.broker.Opt.givenTheta}. Nevertheless, to find the optimal choice of $\theta$ that solves the global problem \eqref{eq.multiagents.broker.Opt.global}, we need to compute the value function $V^\theta_P$ for each $\theta$. This, in turn, requires the knowledge of the equilibrium strategies of the agents that correspond to the fees $\xi^*$ constructed in Proposition \ref{prop:multiagents.optContracts}, as well as the value of the broker in this Stackelberg game. The former is discussed in this section, and the latter is analyzed in Section \ref{se:optimalOF}.

The results of this section hold for an arbitrary fixed $\theta\in\{0,1\}^N$. However, for convenience, we assume that $\cN(\theta) = \{m+1,\cdots, N\}$ with some $0\leq m\leq N-1$.

Notice that, with the fees given by \eqref{fee_optimal}, the objectives of the broker's clients do not depend on their actions nor on the actions of the independent agents. Hence, any equilibrium in the sub-game among the independent agents can trivially be extended to an equilibrium among all agents. This observation is made precise in Theorem \ref{thm:multiagents.main}, and it is only brought up here to explain why it suffices to focus on the equilibria among independent agents, which is the main subject of the remainder of this section.

We begin by noticing that the objective \eqref{eq.sec2.indepObj.def} of an independent agent, by design, is only affected by the actions of the broker and of her clients through the total order flow of the broker's clients, denoted
$$
u:= \sum_{i=m+1}^N \nu^i.
$$
We refer to $u$ as the broker's order flow.
In particular, for the fees $\xi^*$ constructed in Proposition \ref{prop:multiagents.optContracts}, the objective \eqref{eq.sec2.indepObj.def} of an independent agent $i\notin\cN$ can be rewritten as
\begin{align}
& J^{i,\theta}(\nu^i,\nu^{-i})
=\tilde{J}^{i,\theta}(\nu^i,\nu^{m,-i},u)\nonumber \\
&:= \EE\left[\int_0^T X^{i}_t\left(\mu + \sum_{j=1}^m  \lambda_j \nu^j_t + \lambda_0u_t\right)dt -\int_0^T \nu^i_t\left(\sum_{j=1}^m \kappa_j \nu^j_t+\kappa_0u_t\right)dt-\frac{\lambda_i}{2} (X^i_T)^2\right]\nonumber\\
&= \EE\left[\int_0^T X^{i}_t\left(\mu + \sum_{j\neq i,\,j\leq m}  \lambda_j \nu^j_t+\lambda_0u_t\right)dt -\int_0^T \nu^i_t\left(\sum_{j=1}^m \kappa_j \nu^j_t+\kappa_0u_t\right)dt-\frac{\lambda_i}{2} (x^i_0)^2\right],
\label{eq.multiagents.IndepAgents.obj}
\end{align}
where $\nu^{m,-i}$denotes the vector $(\nu^1,\ldots,\nu^m)$ without the $i$th element and $X^i_t=x^i_0 + \int_0^t \nu^i_s ds$.

The main goal of this section is to characterize all Nash equilibria among independent agents who solve 
\begin{equation}\label{eq.multiagents.IndepAgents.game}
\sup_{\tilde\nu^i\in\mathcal{U}}\tilde{J}^{i,\theta}(\tilde\nu^i,\tilde\nu^{-i},u),\quad i=1,\ldots,m,
\end{equation}
for any given order flow of the broker $u \in \mathcal{U}$.

\begin{proposition}\label{prop:Nash.equil}
   For any $u \in \mathcal{U}$, there exists a unique Nash equilibrium $(\tilde\nu^{1,*},\ldots,\tilde\nu^{m,*})$ of \eqref{eq.multiagents.IndepAgents.game}, and it is given by
    \begin{equation}\label{eq.multiagents.tildenu.def}
          \tilde\nu^{i,*}_t =\frac{1}{\kappa_i(m+1)}\left(mY^i_t-\sum_{j\neq i}^m Y^j_t -\kappa_0u_t\right), \quad 1\leq  i \leq m,
    \end{equation}
    where $(Y,Z)$ is the unique solution of the BSDE: 
    \begin{equation}\label{eq.multiagents.ode.direct}
        \begin{split}
                     dY^i_t &= -\left[\mu-\frac{\gamma-\lambda_i/\kappa_i}{m+1}Y^i_t+\sum_{j\neq i}^m\left(\frac{\lambda_j}{\kappa_j}-\frac{\gamma-\lambda_i/\kappa_i}{m+1}\right)Y^j_t+\left(\lambda_0-\frac{\gamma-\lambda_i/\kappa_i}{m+1}\kappa_0\right)u_t\right]dt+Z^i_tdB_t,  \\
            Y^i_T &= 0, \quad 1 \leq i \leq m.
        \end{split}
    \end{equation}
 and
\begin{equation*}
\gamma := \sum_{i=1}^m \frac{\lambda_i}{\kappa_i}.
\end{equation*}
\end{proposition}
\begin{proof}
    Let us fix arbitrary $u\in \mathcal{U}$, $1\leq i \leq m$, $\tilde{\nu}^{-i}\in \mathcal{U}^{m-1}$, and describe an optimal strategy for the agent $i$. Recall that the agent $i$ maximizes the right hand side of \eqref{eq.multiagents.IndepAgents.obj}. We introduce his Hamiltonian:
    \begin{equation*}
    H_t^{i}\left(\tilde\nu^i,\tilde x^i, y^{i},\tilde\nu^{-i}\right)=\tilde \nu^{i} y^{i} + \tilde x^{i}\left(\mu+\lambda_0u_t+\sum_{j\neq i} \lambda_{j} \tilde\nu^{j}\right)-\kappa_{i}\left(\tilde\nu^{i}\right)^{2}-\tilde\nu^i\sum_{j\neq i} \kappa_j\tilde\nu^j-\kappa_0u_t\tilde\nu^i.
    \end{equation*}
    Next, we observe that $H_t^i$ is concave in $(\tilde\nu^i,\tilde x^i)$. Applying the stochastic maximum principle for the i-th agent's problem (see, e.g., Theorem 6.4.6 in \cite{Pham}), we conclude that the strategy defined by
\begin{equation}\label{first.equation.nash.equi.indep}
        \begin{split}
            \tilde\nu^i_t &=\frac{1}{2\kappa_i}\left(Y^i_t -\sum_{j\neq i}^m \kappa_j\tilde\nu^j_t-\kappa_0u_t\right),\\
            dY^i_t &=-\left(\mu + \sum_{j\neq i}^m \lambda_j\tilde\nu^j_t +\lambda_0u_t \right)dt+Z^i_tdB_t,
            \quad Y^i_T = 0,
        \end{split}
\end{equation}      
is optimal.

Moreover, as the objective of the agent $i$ is strictly concave, we conclude that \eqref{first.equation.nash.equi.indep} defines his unique optimal strategy, given $\tilde{\nu}^{-i}\in \mathcal{U}^{m-1}$.
Applying the same argument for every agent $1\leq i \leq m$, we deduce that any solution of the system \eqref{first.equation.nash.equi.indep}, for $i=1,\ldots,m$, defines a Nash equilibrium among the independent agents. By the strict concavity of the individual objectives we obtain that any Nash equilibrium is a solution to \eqref{first.equation.nash.equi.indep}.
Summing up the first equation in \eqref{first.equation.nash.equi.indep} over $i$, we obtain
    \begin{equation}\label{temp.impact.equiv.indep}
        \sum_{j=1}^m \kappa_j \tilde\nu^j_t =\frac{1}{m+1} \left(\sum_{j=1}^mY^j_t -m\kappa_0u_t\right),
    \end{equation}
    and, in turn,
\begin{align*}
&\tilde\nu^i_t =\frac{1}{\kappa_i}\left(Y^i_t -\sum_{j=1}^m \kappa_j\tilde\nu^j_t-\kappa_0u_t\right)
 =  \frac{1}{\kappa_i}\left(Y^i_t - \frac{1}{m+1} \left(\sum_{j=1}^mY^j_t -m\kappa_0u_t\right) -\kappa_0u_t\right),
\quad 1\leq i \leq m.\\
\end{align*}
Plugging the above in the second equation in \eqref{first.equation.nash.equi.indep}, we obtain \eqref{eq.multiagents.ode.direct}.
Thus, we have shown that any Nash equilibrium among the independent agents satisfies \eqref{eq.multiagents.tildenu.def}--\eqref{eq.multiagents.ode.direct}. It remains to notice that \eqref{eq.multiagents.ode.direct} is a standard linear BSDE, and its solution is unique. The latter, in particular, yields uniqueness of the solution to \eqref{first.equation.nash.equi.indep} and hence the uniqueness of equilibrium. 
\qed
\end{proof}

\smallskip

An immediate corollary of Proposition \ref{prop:Nash.equil} is that, with $\mathcal{N}(\theta)=\{m+1,\ldots,N\}$ and with the fees $\xi^*$ given by \eqref{fee_optimal}, the set $\cE(\theta,\xi^*)$ of all equilibria among the agents (see Definition \ref{def:cE}) is given by
\begin{align*}
& \cE(\theta,\xi^*) = \left\{ (\tilde\nu^{1,*}(u),\ldots,\tilde\nu^{m,*}(u),\nu^{m+1},\ldots,\nu^N):\,u=\sum_{i=m+1}^N \nu^i,\,\,\nu^{m+1},\ldots,\nu^N \in \mathcal{U}\right\},
\end{align*}
where $(\tilde\nu^{1,*}(u),\ldots,\tilde\nu^{m,*}(u))$ are given by \eqref{eq.multiagents.tildenu.def}.

\section{Optimization problem of the broker}
\label{se:optimalOF}

As in the previous section, the results of this section hold for an arbitrary fixed $\theta\in\{0,1\}^N$, but, for convenience, we assume that $\cN(\theta) = \{m+1,\cdots, N\}$ with some $0\leq m\leq N-1$.

\smallskip

Herein, we turn to the control problem of the broker. Notice that, with the fees given by \eqref{fee_optimal} and with the strategies of broker's clients denoted by $(\nu^{m+1},\ldots,\nu^N)$, the independent agents will necessarily adapt the strategies $(\tilde\nu^{1,*}(u),\ldots,\tilde\nu^{m,*}(u))$, given by \eqref{eq.multiagents.tildenu.def} with $u=\sum_{i=m+1}^N \nu^i$, and the payoff of the broker can be written as
\begin{align}
&\tilde{J}^\theta_P(u,X^{m+1}_T,\ldots,X^N_T) = \EE\left[\int_0^T \left(X^{0}_t+\sum_{i=m+1}^N x^i_0\right) \left(\mu + \sum_{i=1}^m  \lambda_i\tilde\nu^{i,*}_t+\lambda_0u_t\right)\,dt\label{eq.multiagents.BrokerObj.tildeJ}\right.\\ 
&\left.-\int_0^Tu_t \left(\sum_{i=1}^m \kappa_i\tilde\nu^{i,*}_t+\kappa_0u_t\right)\,dt
- \sum_{i=m+1}^N\frac{\lambda_i}{2} (X^{i}_T)^2
- \sum_{i=m+1}^N R^{i,\theta}\right],\nonumber
\end{align}
where 
\begin{align*}
& X^0_t:=\int_0^t u_tdt,\quad X^i_T = x_0^i + \int_0^T\nu^i_t dt.
\end{align*}
This implies that the value of the broker's objective \eqref{eq.JP.def}, for the fees given by \eqref{fee_optimal}, can be written as
\begin{equation}\label{eq.JP2.def}
J^\theta_P(\xi^*) = \sup_{u\in\mathcal{U}}
\sup_{\substack{X^{m+1}_T,\ldots,X^N_T\in \mathcal{G},\\
\int_0^Tu_tdt = \sum_{i=m+1}^N (X^i_T-x^i_0)}} \tilde{J}^\theta_P(u,X^{m+1}_T,\ldots,X^{N}_T),
\end{equation}
where 
$$ \mathcal{G}:=\{ X\in \mathcal{F}_T, s.t.\text{ } X = \int_0^Tu_tdt, \text{ for some } u \in \mathcal{U}\}. $$

Next, we use \eqref{eq.multiagents.tildenu.def} to deduce
\begin{align*}
& \sum_{i=1}^m \kappa_i \tilde\nu^{i,*}_t = \frac{m}{m+1} \sum_{i=1}^m Y^i_t - \frac{1}{m+1} \sum_{i=1}^m\left(\sum_{j=1}^m Y^j_t - Y^i_t\right) - \frac{m}{m+1} \kappa_0 u_t
=\frac{1}{m+1} \sum_{i=1}^m Y^i_t - \frac{m}{m+1} \kappa_0 u_t,\\
& \sum_{i=1}^m \lambda_i \tilde\nu^{i,*}_t = \frac{m}{m+1} \sum_{i=1}^m \frac{\lambda_i}{\kappa_i} Y^i_t
- \frac{1}{m+1} \sum_{i=1}^m \frac{\lambda_i}{\kappa_i}\left(\sum_{j=1}^m Y^j_t - Y^i_t \right)
- \frac{\gamma}{m+1}\kappa_0 u_t\\
& = \sum_{i=1}^m \frac{\lambda_i}{\kappa_i} Y^i_t
- \frac{\gamma}{m+1} \sum_{i=1}^m Y^i_t
- \frac{\gamma}{m+1}\kappa_0 u_t
= \sum_{i=1}^m \left(\frac{\lambda_i}{\kappa_i} - \frac{\gamma}{m+1}\right) Y^i_t - \frac{\gamma}{m+1}\kappa_0 u_t,
\end{align*}
which leads to
\begin{align*}
& \tilde J^\theta_P(u,X^{m+1}_T,\ldots,X^{N}_T) = \EE \left[\int_0^T \left(X^{0}_t+\sum_{i=m+1}^N x^i_0\right) \left(\mu +\sum_{i=1}^m\left[\frac{\lambda_i}{\kappa_i}-\frac{\gamma}{m+1}\right]Y^i_t + \left[\lambda_0 - \frac{\gamma\kappa_0}{m+1}\right] u_t\right)dt\right.\\
&\left.-\frac{1}{m+1}\int_0^Tu_t\left(\sum_{i=1}^m Y^i_t+\kappa_0 u_t\right)dt
- \sum_{i=m+1}^N\frac{\lambda_i}{2} (X^{i}_T)^2 - \sum_{i=m+1}^N R^{i,\theta} \right].
\end{align*}
where 
$(Y,Z)$ is the unique solution to the linear BSDE \eqref{eq.multiagents.ode.direct}.

\medskip

Let us resolve the optimization over $X^{m+1}_T,\ldots,X^{N}_T$ in \eqref{eq.JP2.def}, for each fixed $u$ and $\omega$. Indeed, the latter amounts to solving the quadratic minimization problem with linear constraints:
\begin{align*}
& \inf_{X^{m+1}_T,\ldots,X^N_T\in\RR} \sum_{i=m+1}^N\frac{\lambda_i}{2} (X^{i}_T)^2,\\
&\text{s.t. }\sum_{i=m+1}^N (X^i_T-x^i_0) = X^0_T,
\end{align*}
where we recall that $X^0_T$ is known given $u$. Constructing the Lagrangian and setting its derivatives to zero, we deduce that the above infimum equals
$$
\frac{1}{2\sum_{i=m+1}^N 1/\lambda_i}\left( X^0_T + \sum_{i=m+1}^N x^i_0\right)^2.
$$
Thus, the broker's objective for a fixed choice of clients $\theta$ can be written as
\begin{equation}\label{eq.JP3.def}
J^\theta_P(\xi^*) = \sup_{u\in\mathcal{U}} \hat{J}^\theta_P(u),
\end{equation}
where
\begin{align}
& \hat J^\theta_P(u) := \EE \left[\int_0^T \left(X^{0}_t+\sum_{i=m+1}^N x^i_0\right) \left(\mu +\sum_{i=1}^m\left[\frac{\lambda_i}{\kappa_i}-\frac{\gamma}{m+1}\right]Y^i_t + \left[\lambda_0 - \frac{\gamma\kappa_0}{m+1}\right] u_t\right)dt\right.\nonumber\\
&\left.-\frac{1}{m+1}\int_0^Tu_t\left(\sum_{i=1}^m Y^i_t+\kappa_0 u_t\right)dt
- \frac{1}{2\sum_{i=m+1}^N 1/\lambda_i}\left( X^0_T + \sum_{i=m+1}^N x^i_0\right)^2 - \sum_{i=m+1}^N R^{i,\theta} \right].
\label{eq.Sec5.hatJ.def}
\end{align}

\medskip

Let us denote $x_0:=\sum_{i=m+1}^N x^i_0$. 
The above expression can be viewed as a backward representation of $\hat{J}^\theta_P(u)$ as it involves $Y$ that solved a BSDE. 
The following lemma establishes a convenient forward representation for $\hat{J}^\theta_P(u)$, which is used in the subsequent analysis.

\begin{lemma}\label{le:1}
For any $u\in\mathcal{U}$, we have 
\begin{equation}
\label{obj_proof_forward}
    \begin{split}
         &\hat{J}^\theta_P(u) = \EE\left[ \int_0^T X^0_t\left(\frac{1}{m+1}\mu + \frac{2}{m+1}\tilde{\gamma}^\top C_t \right)dt
         +\frac{2}{m+1}\int_0^T X^0_t \tilde{\gamma}^\top e^{At} D_t\,dt
         -\frac{2}{m+1} E_T\,D_T \right. \\
         &\left.
         +x_0 \int_0^T \tilde{\gamma}^\top C_t\,dt
         -x_0\tilde{\gamma}^\top \int_0^T e^{At} \,dt\,D_T
         +x_0\tilde{\gamma}^\top \int_0^T e^{At} D_t\, dt\right. \\
         &\left.+\frac{1}{2}\left(\frac{\lambda_0}{m+1}-\frac{2\gamma \kappa_0}{(m+1)^2}-\frac{1}{\sum_{i=m+1}^N 1/\lambda_i}\right)(X_T^0)^2
         +x_0\left(\lambda_0-\frac{\gamma \kappa_0}{m+1}-\frac{1}{\sum_{i=m+1}^N 1/\lambda_i}\right)X_T^0\right.\\
         &\left. 
         -\frac{\kappa_0}{m+1}\int_0^T(u_t)^2dt\right]
         +\left(\mu Tx_0-\frac{1}{2\sum_{i=m+1}^N 1/\lambda_i} x_0^2\right) - \sum_{i=m+1}^N R^{i,\theta},
    \end{split}
\end{equation}
where $A\in \mathbb{R}^{m\times m}$, $ \mathbf{b}, \tilde{\gamma}\in \mathbb{R}^m$ are defined by 
\begin{equation}\label{notation.matrix}
    \begin{split}
         A_{ij} := \left\{\begin{array}{cc}
    \frac{1}{m+1}(\gamma-\lambda_i/\kappa_i),  & i=j, \\
      -\lambda_j/\kappa_j+\frac{1}{m+1}(\gamma-\lambda_i/\kappa_i), & i\neq j, \\
    \end{array}\right.
    \end{split}
\end{equation}
 \begin{equation}
    \begin{split}\label{param.state}
     \mathbf{b}&:=\left(\frac{\left(\gamma-\lambda_1/\kappa_1\right)\kappa_0}{m+1}-\lambda_0,\ldots,\frac{\left(\gamma-\lambda_m/\kappa_m\right)\kappa_0}{m+1}-\lambda_0\right)^\top,\\
     \tilde{\gamma}&:=\left(\frac{\lambda_1}{\kappa_1}-\frac{\gamma}{m+1},\ldots,\frac{\lambda_m}{\kappa_m}-\frac{\gamma}{m+1}\right)^\top, \\
     \end{split}
 \end{equation}
     and
 \begin{equation}\label{extra.state}
        \begin{split}
           C_t:= \mu\,e^{At}\int_t^Te^{-As}\mathbf{1}ds,\quad
          \quad  D_t:= \int_0^t u_s\,e^{-As}\mathbf{b}\,ds, \quad E_t:= \int_0^t X^0_s \,\tilde{\gamma}^\top e^{As}\,ds.
        \end{split}
    \end{equation}
\end{lemma}

\begin{proof}
Integrating by parts and recalling \eqref{eq.multiagents.ode.direct}, we obtain:
\begin{align*}
& -\frac{1}{m+1} \EE \int_0^T u_t \sum_{i=1}^m Y^i_t \,dt
= -\frac{1}{m+1} \EE \int_0^T X^0_t \left[m \mu - \sum_{i=1}^m \frac{\gamma-\lambda_i/\kappa_i}{m+1}Y^i_t\right.\\
&\left. + (m-1)\sum_{j=1}^m \frac{\lambda_j}{\kappa_j} Y^j_t 
- \frac{m\gamma-\sum_{i=1}^m\lambda_i/\kappa_i}{m+1}\sum_{j=1}^m Y^j_t
+ \sum_{i=1}^m \frac{\gamma-\lambda_i/\kappa_i}{m+1} Y^i_t 
+  \left(m \lambda_0 - \sum_{i=1}^m \frac{\gamma-\lambda_i/\kappa_i}{m+1}\kappa_0\right)u_t\right]dt\\
&= -\frac{1}{m+1} \EE \int_0^T X^0_t \left[m \mu 
+ (m-1)\sum_{j=1}^m \frac{\lambda_j}{\kappa_j} Y^j_t 
- \frac{m\gamma-\gamma}{m+1}\sum_{j=1}^m Y^j_t
+  \left(m \lambda_0 - \frac{m\gamma-\gamma}{m+1}\kappa_0\right)u_t\right]dt.
\end{align*}
Plugging the above into the right hand side of \eqref{eq.Sec5.hatJ.def}, we obtain:
\begin{align}
&\hat{J}^\theta_P(u) =\EE \left[\int_0^TX^0_t\left(\frac{1}{m+1}\mu + \frac{2}{m+1}\sum_{i=1}^m\left[\frac{\lambda_i}{\kappa_i}-\frac{\gamma}{m+1}\right]Y^i_t\right)dt\right. \nonumber\\
&\left.+\frac{1}{2}\left(\frac{\lambda_0}{m+1}-\frac{2\gamma \kappa_0}{(m+1)^2}-\frac{1}{\sum_{i=m+1}^N1/\lambda_i}\right)(X_T^0)^2+x_0\left(\lambda_0-\frac{\gamma \kappa_0}{m+1}-\frac{1}{\sum_{i=m+1}^N 1/\lambda_i}\right)X_T^0 \right. 
\nonumber\\
&\left.+x_0\sum_{i=1}^m\int_0^T\left(\frac{\lambda_i}{\kappa_i}-\frac{\gamma}{m+1}\right)Y^i_tdt-\frac{\kappa_0}{m+1}\int_0^T(u_t)^2dt\right]+\left(\mu Tx_0-\frac{1}{2\sum_{i=m+1}^N 1/\lambda_i}x_0^2\right)- \sum_{i=m+1}^N R^{i,\theta}.     \label{final.obj.broker}
\end{align}

\smallskip

Next, we recall that the linear BSDE \eqref{eq.multiagents.ode.direct} has a semi-explicit solution: 
\begin{equation}\label{exp.sol.FBSDE}
    Y_t=\mu e^{At}\int_t^Te^{-As}\mathbf{1}ds-\EE\left(e^{At}\int_t^Tu_se^{-As}\mathbf{b}ds  \left|\right. \mathcal{F}_t\right).
\end{equation} 
Plugging the above expression into \eqref{final.obj.broker} and recalling the definition of $C_t$ in \eqref{extra.state}, we obtain
\begin{align*}
         &\hat{J}^\theta_P(u) = \EE\left[ \int_0^T X^0_t\left(\frac{1}{m+1}\mu +
          \frac{2}{m+1}\tilde{\gamma}^\top C_t-\frac{2}{m+1}\tilde{\gamma}^\top \EE\left(e^{At}\int_t^T u_se^{-As}\mathbf{b}\,ds  \vert \mathcal{F}_t\right) \right)dt \right. \\
         &\left.+x_0\tilde{\gamma}^\top \int_0^T C_t\,dt
         -x_0\tilde{\gamma}^\top \int_0^T\EE\left(e^{At}\int_t^Tu_se^{-As}\mathbf{b}\,ds  \left|\right. \mathcal{F}_t\right)dt\right.\\
        &\left.+\frac{1}{2}\left(\frac{\lambda_0}{m+1}-\frac{2\gamma \kappa_0}{(m+1)^2}-\frac{1}{\sum_{i=m+1}^N 1/\lambda_i}\right)(X_T^0)^2
        +x_0\left(\lambda_0-\frac{\gamma \kappa_0}{m+1} - \frac{1}{\sum_{i=m+1}^N 1/\lambda_i}\right)X_T^0\right.\\
         &\left. -\frac{\kappa_0}{m+1}\int_0^T(u_t)^2dt\right]
     +\left(\mu Tx_0-\frac{1}{2\sum_{i=m+1}^N 1/\lambda_i}x_0^2\right) - \sum_{i=m+1}^N R^{i,\theta}.
\end{align*}
Using Fubini's theorem and the tower property, we remove the conditional expectations in the right hand side of the above.
Finally, noticing that 
\begin{align*}
\int_t^T u_se^{-As}\mathbf{b}\,ds = D_T-D_t
\end{align*}
and recalling the definition of $E_T$ (in \eqref{extra.state}), we obtain the statement of the lemma.
\qed
\end{proof}

\medskip

Next, we introduce our main assumption.
\begin{ass}\label{ass:main}
The parameters $\{\lambda_i,\kappa_i\}$ are such that
\begin{equation}
2\|\tilde{\gamma}\|\, \|\mathbf{b}\|\,e^{\|A\|_{2}T}\,T^2 + \lambda_0 T/2 < \kappa_0,
\end{equation}
where $\|\cdot\|_2$ is the 2-norm of a matrix, $\|.\|$ denotes the Euclidean norm in $\mathbb{R}^m$, and we recall
\begin{align*}
&\tilde{\gamma}=\left(\frac{\lambda_1}{\kappa_1}-\frac{\gamma}{m+1},\ldots,\frac{\lambda_m}{\kappa_m}-\frac{\gamma}{m+1}\right)^\top,      \quad \mathbf{b}=\left(\frac{\left(\gamma-\lambda_1/\kappa_1\right)\kappa_0}{m+1}-\lambda_0,\ldots,\frac{\left(\gamma-\lambda_m/\kappa_m\right)\kappa_0}{m+1}-\lambda_0\right)^\top,\\
&\gamma = \sum_{i=1}^m \frac{\lambda_i}{\kappa_i},
\quad A_{ij} := \left\{\begin{array}{cc}
    \frac{1}{m+1}(\gamma-\lambda_i/\kappa_i),  & i=j, \\
      -\lambda_j/\kappa_j+\frac{1}{m+1}(\gamma-\lambda_i/\kappa_i), & i\neq j, \\
    \end{array}\right.
\end{align*}
\end{ass}

When verifying the above assumption, it is convenient to recall that $\|A\|_2\leq \|A\|_{\text{Frob}}$, where
\begin{align*}
\|A\|_{\text{Frob}}:=\frac{1}{m+1}\left(\sum_{i=1}^m\left(\gamma-\frac{\lambda_i}{\kappa_i}\right)^2+ \sum_{i=1}^m\sum_{j\neq i}^m \left(\frac{\lambda_j}{\kappa_j}(m+1)-\gamma+\frac{\lambda_i}{\kappa_i}\right)^2\right)^{1/2}.
\end{align*}

\smallskip

The next proposition shows that the broker's optimization problem \eqref{eq.JP3.def} is well-posed under Assumption \ref{ass:main}.

\begin{proposition}\label{prop.unique.sol}
Under Assumption \ref{ass:main}, there exists a unique maximizer $u^*$ of $\hat{J}^\theta_P(\cdot)$ over $\mathcal{U}$. Moreover, the optimal strategy $u^*$ is deterministic. 
\end{proposition}
\begin{proof}
Using Lemma \ref{le:1}, we deduce that $\hat{J}^\theta_P(u) = \EE G^{u}$, where
\begin{equation*}
        \begin{split}
        G^{u} &:=  \int_0^T X^0_t\left(\frac{1}{m+1}\mu+\frac{2}{m+1}\tilde{\gamma}^\top C_t \right)dt
         +\frac{2}{m+1}\int_0^T X^0_t \tilde{\gamma}^\top e^{At} D_t\,dt
         -\frac{2}{m+1} E_T\,D_T \\
         &
         +x_0 \int_0^T \tilde{\gamma}^\top C_t\,dt
         -x_0\tilde{\gamma}^\top \int_0^T e^{At} \,dt\,D_T
         +x_0\tilde{\gamma}^\top \int_0^T e^{At} D_t\, dt \\
         &+\frac{1}{2}\left(\frac{\lambda_0}{m+1}-\frac{2\gamma \kappa_0}{(m+1)^2}-\frac{1}{\sum_{i=m+1}^N1/\lambda_i}\right)(X_T^0)^2
         +x_0\left(\lambda_0-\frac{\gamma \kappa_0}{m+1}-\frac{1}{\sum_{i=m+1}^N1/\lambda_i}\right)X_T^0\\
         & 
         -\frac{\kappa_0}{m+1}\int_0^T(u_t)^2dt
         +\left(\mu Tx_0-\frac{1}{2\sum_{i=m+1}^N1/\lambda_i}x_0^2\right) - \sum_{i=m+1}^N R^{i,\theta}.
    \end{split}
    \end{equation*}
Next, we observe: 
\begin{equation}
    \begin{split}
        \frac{2}{m+1}\left|D_TE_T -\int_0^T X^0_t \tilde{\gamma}^\top e^{At} D_t\,dt \right| &\leq
        \frac{2}{m+1} \left(\int_0^T \left| X^0_t\tilde{\gamma} e^{At}\int_t^Te^{-As}\mathbf{b}\,u_s\,ds \right| dt \right) \\
        &\leq \frac{2}{m+1}\,\|\tilde{\gamma}\|\,\|\mathbf{b}\|\,e^{\|A\|_2T}\, \left(\int_0^T |X^0_t|dt \int_0^T |u_t|dt\right)\\ 
        &\leq\frac{2}{m+1}\,\|\tilde{\gamma}\|\,\|\mathbf{b}\|\,e^{\|A\|_2T}\,T\, \left(\int_0^T |u_t| dt \right)^2\\
        &\leq\frac{2}{m+1}\,\|\tilde{\gamma}\|\,\|\mathbf{b}\|\,e^{\|A\|_2T}\,T^2\, \int_0^T u_t^2 dt
    \end{split}
\end{equation}
where we used Jensen's inequality:
\begin{align*}
\left(\int_0^T |u_t| dt \right)^2\leq T\,\int_0^T u_t^2 dt.
\end{align*}
The above inequality also yields
\begin{equation*}
(X^0_T)^2 \leq T\, \int_0^T (u_t)^2dt.
\end{equation*}
Collecting the above, we conclude that 
\begin{equation*} 
    \begin{split}
   & G^{u} \leq \left(\frac{2}{m+1}\|\tilde{\gamma}\| \|\mathbf{b}\|e^{\|A\|_{2}T}T^2 - \frac{\kappa_0}{m+1} + \frac{\lambda_0 T}{2(m+1)}\right)\int_0^T(u_t)^2dt \\
    &+ x_0\tilde{\gamma}\int_0^TC_tdt - x_0\tilde{\gamma}\int_0^Te^{At}dt\,D_T+x_0\tilde{\gamma}\int_0^Te^{At}D_tdt \\
   &+ \int_0^T X^0_t\left(\frac{1}{m+1}\mu+\frac{2}{m+1}\tilde{\gamma}C_t\right)dt +x_0\left(\lambda_0-\frac{\gamma \kappa_0}{m+1}-\frac{1}{\sum_{i=m+1}^N 1/\lambda_i}\right)X_T^0+\tilde a, 
   \end{split}
\end{equation*}
where $\tilde a\in \RR$.
Notice that the second and the third lines in the right hand side of the above display are linear in $u$. In addition, Assumption \ref{ass:main} yields that the coefficient in front of $\int_0^T(u_t)^2 dt$ is strictly negative, which implies that the above expression is strictly concave as a function of $u\in\mathcal{U}_d$, where we introduced the set of deterministic strategies $\mathcal{U}_d:=L^2([0,T])$. As $G^u$ is linear-quadratic in $u$, we conclude that it is also strictly concave (this can deduced easily by contradiction), which yields the statement of the proposition.
\qed
\end{proof}


\medskip

The above lemma shows that there is no loss of optimality in reducing the optimization problem \eqref{eq.JP3.def} of the broker to the deterministic set of strategies $u\in\mathcal{U}_d = L^2([0,T])$. As the objective $\hat J_P^\theta$ is linear-quadratic, we can find its maximizer by setting to zero its derivative.

\begin{lemma}\label{le:frechet}
The mapping $u\mapsto\hat J_P^\theta(u)$ is Fréchet-differentiable w.r.t. the $L^2$-norm on $\mathcal{U}_d$, and its derivative is given by the following linear functional of $v\in \mathcal{U}_d$: 
\begin{align*}
    D \hat J^\theta_P(u)(v)&:= \int_0^Tv_t\left(\frac{\mu(T-t)}{m+1}+p_t-\frac{2}{m+1}\tilde{\gamma}q_t+x_0\left[\lambda_0-\frac{\gamma\kappa_0}{m+1}-\frac{1}{\sum_{j=1}^N 1/\lambda_j}\right] - x_0\tilde{\gamma}r_t-\frac{2\kappa_0}{m+1}u_t\right)dt\\
\label{derivative.def}
\end{align*}
where 
$(q,p,r)\in C([0,T],\mathbb{R}\times\mathbb{R}^m\times\mathbb{R}^m) $ is the unique solution of the (non-coupled) system of ODE:
\begin{equation}\label{eq.pqr.def}
    \begin{split}
        dp_t &= -\frac{2}{m+1}\tilde{\gamma}^{\top}Y_tdt,\quad p_T = \left(\frac{\lambda_0}{m+1}-\frac{\gamma \kappa_0}{(m+1)^2}-\frac{1}{\sum_{j=m+1}^N1/\lambda_j}\right)X^0_T, \\
        dq_t &= (-Aq_t+X^0_t\mathbf{b})dt,\quad q_0 = \mathbf{0}, \\
        dr_t &= (-Ar_t+\mathbf{b})dt,\quad r_0 = \mathbf{0},
     \end{split}
\end{equation}
where $\mathbf{b}$, $\tilde{\gamma}$ and $A$ are defined in \eqref{notation.matrix}--\eqref{param.state}.
\end{lemma}
\begin{proof}
For any $u\in \mathcal{U}_d$ we define $X_t^{0,u}:=\int_0^t u_tdt$, and we introduce $Y^{u}$ defined as the unique solution of the ODE: 
\begin{equation}\label{eq.Y_2}
    \begin{split}
        dY_t &=( A Y_t + \mathbf{b}u_t-\mu \mathbf{1})dt,\\
        Y_T &= \mathbf{0},
    \end{split}
\end{equation}
where $\mathbf{1}=(1,\cdots,1)^T\in \mathbb{R}^m$. Recall \eqref{eq.Y_2} can be written as 
\begin{equation}
    Y_t = \mu e^{At}\int_t^Te^{-As}\mathbf{1}ds-e^{At}\int_t^Tu_se^{-As}\mathbf{b}ds.
\end{equation}

Let $u,v \in \mathcal{U}_d$. Then,
\begin{equation*}
    \begin{split}
    &\hat J^\theta_P(u+v)-\hat J^\theta_P(u) 
    = \frac{\mu}{m+1}\int_0^TX^{0,v}_sds +\frac{2}{m+1} \int_0^T X^{0,u}_s\left(\sum_{i=1}^m(\frac{\lambda_i}{\kappa_i}-\frac{1}{m+1}\gamma)(Y^{i,u+v}_s-Y_s^{i,u}) \right)ds\\
    &+\frac{2}{m+1}\int_0^T X^{0,v}_s\left(\sum_{i=1}^m(\frac{\lambda_i}{\kappa_i}-\frac{1}{m+1}\gamma) \left(Y^{i,u}_s-e^{As}\int_s^Te^{-Ar}\mathbf{b}v_rdr\right) \right)ds
    -\frac{2\kappa_0}{m+1}\int_0^Tu_sv_sds\\
    &\left(\frac{\lambda_0}{m+1}-\frac{2\gamma \kappa_0}{(m+1)^2}-\frac{1}{\sum_{i=m+1}^N\frac{1}{\lambda_i}}\right)X^{0,u}_TX^{0,v}_T +\int_0^T X^{0,v}_s\left(\sum_{i=1}^m(\frac{\lambda_i}{\kappa_i}-\frac{1}{m+1}\gamma))Y^{i,v}_s \right)ds-\frac{\kappa_0}{m+1}\int_0^T(v_s)^2ds \\
    &\frac{1}{2}\left(\frac{\lambda_0}{m+1}-\frac{2\gamma \kappa_0}{(m+1)^2}-\frac{1}{\sum_{i=m+1}^N\frac{1}{\lambda_i}}\right)(X^{0,v}_T)^2+x_0\sum_{i=1}^m\int_0^T\left(\frac{\lambda_i}{\kappa_i}-\frac{\gamma}{m+1}\right)(Y^{i,u+v}_t-Y^{i,u}_t)dt \\
    &+x_0\left(\lambda_0-\frac{\gamma \kappa_0}{m+1} -\frac{1}{\sum_{i=m+1}^N\frac{1}{\lambda_i}}\right)X_T^{0,v}.
    \end{split}
\end{equation*}
It is a standard exercise to check that the linear part (in $u$) of the right hand side of the above gives the desired Fréchet derivative:
\begin{equation}\label{frechet.derivative}
\begin{split}
&D \hat J^\theta_P(u)(v):= \frac{\mu}{m+1}\int_0^TX^{0,v}_sds +\frac{2}{m+1} \int_0^T X^{0,u}_s\left(\sum_{i=1}^m(\frac{\lambda_i}{\kappa_i}-\frac{1}{m+1}\gamma)(Y^{i,u+v}_s-Y_s^{i,u}) \right)ds\\
    &+\frac{2}{m+1}\int_0^T X^{0,v}_s\left(\sum_{i=1}^m(\frac{\lambda_i}{\kappa_i}-\frac{1}{m+1}\gamma))Y^{i,u}_s \right)ds
   +\left(\frac{\lambda_0}{m+1}-\frac{2\gamma \kappa_0}{(m+1)^2}-\frac{1}{\sum_{i=m+1}^N\frac{1}{\lambda_i}}\right)X^{0,u}_TX^{0,v}_T\\
    & -\frac{2\kappa_0}{m+1}\int_0^Tu_sv_sds+x_0\sum_{i=1}^m\int_0^T\left(\frac{\lambda_i}{\kappa_i}-\frac{\gamma}{m+1}\right)(Y^{i,u+v}_t-Y^{i,u}_t)dt +x_0\left(\lambda_0-\frac{\gamma \kappa_0}{m+1} -\frac{1}{\sum_{i=m+1}^N\frac{1}{\lambda_i}}\right)X_T^{0,v}.
    \end{split}
\end{equation}

\medskip

Next, we rewrite $DJ^\theta_P(u)$ in a more convenient way. To this end, we observe:
\begin{equation*}
    \begin{split}
    &\int_0^T X^{0,u}_t\sum_{i=1}^m(\frac{\lambda_i}{\kappa_i}-\frac{\gamma}{m+1})(Y^{i,u+v}_t-Y^{i,u}_t) dt= -\int_0^T X^{0,u}_t\tilde{\gamma}^T\int_t^\top e^{A(t-s)}\mathbf{b}v_sdsdt \\
    &= -\int_0^Tv_t\tilde{\gamma}^T \int_0^t e^{-A(t-s)}\mathbf{b} X^{0,u}_sdsdt, \\
    &\int_0^T X^{0,v}_t\left(\sum_{i=1}^m(\frac{\lambda_i}{\kappa_i}-\frac{\gamma}{m+1})Y^{i,u}_t \right)dt = \int_0^T \left(\sum_{i=1}^m(\frac{\lambda_i}{\kappa_i}-\frac{\gamma}{m+1})Y^{i,u}_t \right)\int_0^tv_sdsdt = \int_0^Tv_t \int_t^T\tilde{\gamma}^\top Y^u_sdsdt, \\
    &X^{0,u}_TX^{0,\eta}_T = \int_0^Tv_t(X^{0,u}_T)dt, \\
    &\int_0^T(Y_t^{u+v}-Y^u_t)dt = -\int_0^T e^{At}\int_t^Tv_se^{-As}\mathbf{b}dsdt = -\int_0^Tv_t\int_0^te^{-A(t-s)}\mathbf{b}dsdt.
    \end{split}
\end{equation*}
Using the above, we deduce:
\begin{equation}
    \begin{split}
        DJ^\theta_P(u)(v)&:=\int_0^Tv_t\left(\frac{\mu(T-t)}{m+1}+\frac{2}{m+1}\tilde{\gamma}^\top\int_t^TY^u_sds-\frac{2}{m+1}\tilde{\gamma}^Te^{-At}\int_0^te^{As}X_s^{0,u}\,\mathbf{b}\,ds\right.\\
&\left.+\left[\frac{\lambda_0}{m+1}-\frac{2\gamma \kappa_0}{(m+1)^2}-\frac{1}{\sum_{i=m+1}^N\frac{1}{\lambda_i}}\right]X_T^{0,u}+x_0\left[\lambda_0-\frac{\gamma\kappa_0}{m+1}-\frac{1}{\sum_{j=m+1}^N\frac{1}{\lambda_j}}\right]\right.\\
&\left.-x_0\tilde{\gamma}^\top e^{-At}\int_0^te^{As}\,\mathbf{b}\,ds-\frac{2\kappa_0}{m+1}u_t\right)dt
    \end{split}
\end{equation}
Recalling \eqref{eq.pqr.def}, we obtain the statement of the lemma.
\qed
\end{proof}

\medskip

The above lemma allows us to characterize the optimal order flow of the broker in terms of the unique solution of a linear forward-backward system of ODEs, arising as a combination of \eqref{eq.pqr.def} and
\begin{equation}\label{eq.multiagents.ustar.def}
        \begin{split}
        &dY_t = (AY_t + \mathbf{b} u_t)dt,\quad Y_T = \mathbf{0}, \\
        &dX^0_t =\left[\frac{\mu(T-t)}{2\kappa_0}+\frac{(m+1)}{2\kappa_0}p_t-\frac{1}{\kappa_0}\tilde{\gamma}^\top q_t+x_0\frac{m+1}{2\kappa_0}\left(\lambda_0-\frac{\gamma\kappa_0}{m+1}-\frac{1}{\sum_{i=m+1}^N\frac{1}{\lambda_i}}\right)\right.\\
        &\left.{\color{red}-}\frac{m+1}{2\kappa_0}x_0\tilde{\gamma}^\top r_t\right] dt,\quad X^0_0=0.
    \end{split}
    \end{equation}

\begin{proposition}\label{prop:broker.opt}
    Under Assumption \ref{ass:main}, there exists a unique classical solution $(p,q,r,Y,X^0)$ to \eqref{eq.pqr.def}, \eqref{eq.multiagents.ustar.def}, and the optimal order flow $u^*_t$ of the broker is equal to $dX^0_t/dt$. 
\end{proposition}

\begin{proof}
By Proposition \ref{prop.unique.sol}, we know that there exits a unique optimal control $u^*\in \mathcal{U}_d$. Then, Lemma \ref{le:frechet} implies that $X^{0,*}_t := \int_0^t u^*_s ds$ satisfies the second line of \eqref{eq.multiagents.ustar.def} for a.e. $t$, with $(p^*,q^*,r^*,Y^*)$ defined via \eqref{eq.pqr.def} and the first line of \eqref{eq.multiagents.ustar.def}. Noticing that $(p,q,r)$ are continuous, we conclude that $X^0$ is continuously differentiable and that it satisfies the second line of \eqref{eq.multiagents.ustar.def} for all $t$. Finally, for any solution $(p,q,r,Y,X^0)$ to \eqref{eq.pqr.def}, \eqref{eq.multiagents.ustar.def}, the Frechet derivative of $\hat J_P^\theta$ at $u_t:=dX^0_t/dt$ is zero, which implies that $u=u^*$ and in turn that $(p,q,r,Y)=(p^*,q^*,r^*,Y^*)$.
\qed
\end{proof}

\medskip

The following theorem summarizes all the results we have established.

\begin{theorem}\label{thm:multiagents.main}
For any $\theta\in\{0,1\}^N$ and any $\{R^{i,\theta}\}\in\RR^N$, the set of fees $\xi^*=\{\xi^{i,*}\}_{i\in\mathcal{N}(\theta)}$ given by \eqref{fee_optimal} is optimal for the broker's local optimization problem \eqref{eq.multiagents.broker.Opt.givenTheta}. Provided the broker chooses this set of fees, the following holds.

\begin{itemize}
\item 
Any choice of equilibrium strategies $\{\nu^{i,*}\}_{i=1}^N\in\mathcal{E}(\theta,\xi^*)$ that is optimal for the broker has the following structure: the strategies of independent agents, $\{\nu^{i,*}\}_{i\notin\mathcal{N}(\theta)}$, are determined uniquely by \eqref{eq.multiagents.tildenu.def}--\eqref{eq.multiagents.ode.direct}, and the clients' strategies $\{\nu^{i,*}\}_{i\in\mathcal{N}(\theta)}$ can be chosen arbitrarily subject to 
$$
\sum_{i\in\mathcal{N}(\theta)} \nu^{i,*} = u^*,\quad X^{i,*}_T=x^i_0+\int_0^T \nu^{i,*}_t dt=x^i_0+\frac{2}{\lambda_i\sum_{j=m+1}^N \frac{1}{\lambda_j}}\left(\int_0^Tu^*_t dt\right),
$$
where $u^*$ is defined in Proposition \ref{prop:broker.opt}. 

\item The broker's value $V^{\theta}_P$ for a given $\theta$, defined in \eqref{eq.multiagents.broker.Opt.givenTheta}, satisfies
\begin{equation}\label{eq.multiagent.broker.Vtheta.formula}
V^{\theta}_P= J^\theta_P(\xi^*) = \hat J^\theta_P(u^*),
\end{equation}
where $\xi^*=\{\xi^{i,*}\}_{i\in\mathcal{N}(\theta)}$ is given by \eqref{fee_optimal}, $\hat J^\theta_P$ is defined in \eqref{eq.Sec5.hatJ.def}, and $u^*$ is defined in Proposition \ref{prop:broker.opt}.

\item In any equilibrium $\{\nu^{j,*}\}_{j=1}^N\in\mathcal{E}(\theta,\xi^*)$, the value $V^{i,\theta}$ of the optimization problem \eqref{eq.sec2.indepOpt.def} of each independent agent $i\notin\mathcal{N}(\theta)$ is given by the right hand side of \eqref{eq.multiagents.IndepAgents.obj}, with $\{\nu^{j}\}_{j=1}^N$ replaced by $\{\nu^{j,*}\}_{j=1}^N$.
In addition, $V^{i,\theta}$ is the same for any choice of $\{\nu^{j,*}\}_{j=1}^N\in\mathcal{E}(\theta,\xi^*)$ that is optimal for the broker.
\end{itemize}
\end{theorem}


\medskip

Let us now propose an endogenous definition of the clients' reservation values.
Notice that these reservation values are needed in order to determine the broker's value $V^\theta_P$ (via \eqref{eq.multiagent.broker.Vtheta.formula} and \eqref{eq.Sec5.hatJ.def}) 
The latter, in turn, is needed to determine which $\theta$ is optimal.

\begin{definition}
\label{def:endog.R}
For every $\theta\in\{0,1\}^N$ and every $i\in\mathcal{N}(\theta)$ we define
\begin{equation}\label{eq.Rtheta.def}
R^{i,\theta} := V^{i,\theta'},
\end{equation}
where all entries of $\theta'$ are equal to those of $\theta$, except for the $i$-th entry which is equal to zero, and $V^{i,\theta'}$ is given by the right hand side of \eqref{eq.multiagents.IndepAgents.obj}, with $\theta$ replaced by $\theta'$ and with a choice of equilibrium strategies $\{\nu^{j}\}_{j=1}^N\in\mathcal{E}(\theta',\xi^*)$ that is optimal for the broker (see Theorem \ref{thm:multiagents.main} and note that $i\notin\mathcal{N}(\theta')$).
\end{definition}

Note that the above definition is consistent, in the sense that the right hand side of \eqref{eq.Rtheta.def} does not depend on $\{R^{i,\theta'}\}_i$. Indeed, the right hand side of \eqref{eq.multiagents.IndepAgents.obj} depends only on $\{\nu^{j}\}_{j\notin\cN(\theta')}$ and $u=\sum_{j\in\cN(\theta')}\nu^j$. In any equilibrium that is optimal for the broker, the latter quantities are determined uniquely by \eqref{eq.multiagents.tildenu.def}--\eqref{eq.multiagents.ode.direct} and by Proposition \ref{prop:broker.opt}, and they do not depend on $\{R^{i,\theta'}\}_i$.

\smallskip

The motivation for the above definition of a reservation value is clear. Namely, each client of the broker has an opportunity to trade directly in the market (i.e., to become an independent agent), hence, the reservation value of the client must equal the maximum objective value he can achieve by such trading.

\medskip

The above definition and Theorem \ref{thm:multiagents.main} give us a method for computing $V^{\theta}_P$, for each $\theta$. Then, an optimal $\theta$ can be found by maximizing $V^{\cdot}_P$. The latter is accomplished by an exhaustive search, in the next section, which is realistic for small $N$ or if the choices of $\theta$ are restricted to a small enough subset of $\{0,1\}^N$. 

\section{Numerical experiments}
\label{se:numerics}

In this section we describe five numerical simulations to study the structure of an optimal clients' portfolio for the broker, as well as the dependence of the values (i.e., the expected equilibrium profits) of the broker and of the agents on the price impact parameters $\{\kappa_i,\lambda_i\}$. In particular, we ask: does the broker include all agents in her optimal portfolio? does the broker prefer agents will low or high price impact coefficients? do the agents benefit from the presence of the broker?
In all experiments, we fix $\mu=1$, $T=1$, $x^i_0=0$.

\subsection{Two agents, dependence of broker's value on $\lambda$}
    
In this subsection, we set $N=2$, fix $\kappa_0=\kappa_1 = \kappa_2 = 10^{-1}$, $\lambda_0=10^{-3}$, and consider the dependence of the broker's value on $(\lambda_1,\lambda_2)$. We generate $M=100$ equidistant values $(\lambda_1,\lambda_2) \in [10^{-3},5\times 10^{-2}]\times[10^{-3},5\times 10^{-2}]$ and show the value of the broker in Figure \ref{fig:1}. 
In this experiment, the broker optimally takes both agents as her clients (i.e., $\theta^*=(1,1)$), for every value of $(\lambda_1,\lambda_2)$.

We observe that the broker's value is larger for small $(\lambda_1,\lambda_2)$. This can be explained as follows. Each agent benefits from the permanent impact of other agents, as they trade in the same directions (because they have the same initial inventories and observe the same signal $\mu$). In particular, a large value of $\lambda_1$ makes the second agent more optimistic about his profits in case he decides to trade directly in the market. The latter means that the second agent has larger reservation value, for which he needs to be compensated by the broker, which in turn reduces the profits of the broker from taking the first agent as a client. Analogous argument applies when $\lambda_2$ is large.

\begin{figure}[H]
        \centering
        \includegraphics[width=0.5\textwidth]{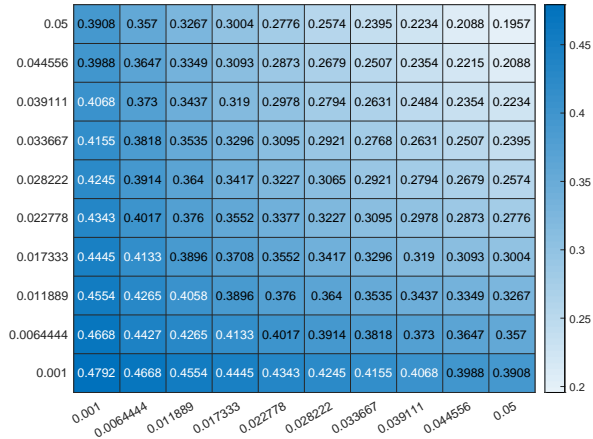}
        \caption{Broker's value as a function of $\lambda_1$ (horizontal axis) and $\lambda_2$ (vertical axis).}
        \label{fig:1}
\end{figure}

\subsection{Two agents, dependence of broker's value on $\kappa$}
\label{subse:numerics.2}
    
In this subsection, we set $N=2$, fix $\lambda_0=\lambda_1 = \lambda_2 =10^{-2}$, $\kappa_0=10^{-2}$, and consider the dependence of the broker's value on $(\kappa_1,\kappa_2)$. We generate $M=100$ equidistant values $(\kappa_1,\kappa_2) \in [10^{-2},10^{-1}]\times[10^{-2},10^{-1}]$ and show the value of the broker in Figure \ref{fig:2}. 
In this experiment, the broker optimally takes both agents as her clients (i.e., $\theta^*=(1,1)$), for every value of $(\kappa_1,\kappa_2)$.

We observe that the broker's value is larger for large $(\kappa_1,\kappa_2)$. This can be explained by the fact that large $\kappa_i$ reduces the value of agent $i$ in case he decides to trade directly in the market, thus reducing his reservation value and in turn increasing the broker's profit from taking this agent as a client.

\begin{figure}[H]
        \centering
        \includegraphics[width=0.5\textwidth]{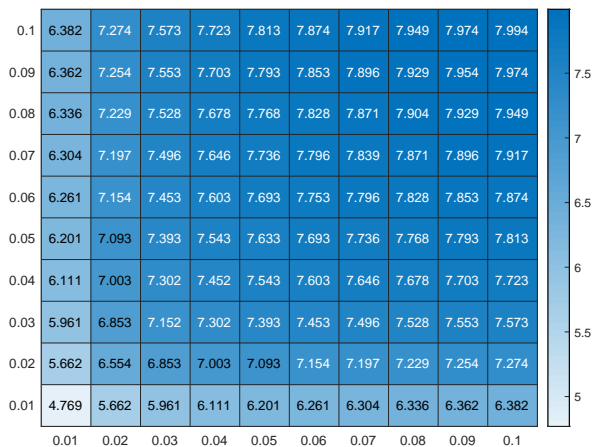}
        \caption{Broker's value as a function of $\kappa_1$ (horizontal axis) and $\kappa_2$ (vertical axis).}
        \label{fig:2}
\end{figure}

\subsection{Many agents, choice of portfolio via $\lambda$}
   
In this subsection, we set $N=100$, fix $\kappa_i= 5\times10^{-2}$, $\lambda_0 =10^{-4}$, $5\times \kappa_0=10^{-2}$, generate $\{\lambda_i\}_{i=1}^N$ as independent realizations of a uniform random variable on $(10^{-4},10^{-3})$, and consider the dependence of the broker's local value function $V^{\theta}_P$ on her portfolio of clients $\theta\in\{0,1\}^N$.
As it is too computationally expensive to compute $V^{\theta}_P$ for all $\theta\in\{0,1\}^N$, we restrict the analysis to the portfolios $\theta$ chosen as the $100p$-percentile of agents with the highest or lowest $\lambda_i$: e.g., $\theta$ may be such that $\theta_i=1$ if and only if $\lambda_i$ belongs to the bottom $10\%$ of $\{\lambda_i\}$. The local value of the broker for $p\in [0,1]$ is shown in Figure \ref{fig:3}.

Figure \ref{fig:3} shows that it is not optimal for the broker to take too many clients: in fact, it is better to take too few than too many. It is also worth mentioning that the broker is almost indifferent between choosing the agents with high or low $\lambda$, which indicates that permanent impact coefficient is not a good metric for choosing a portfolio of clients.

\begin{figure}[H]
        \centering
        \includegraphics[width=0.5\textwidth]{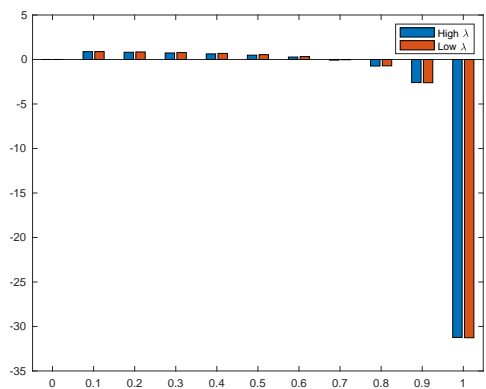}
        \caption{Broker's local value $V^{\theta}_P$ when she chooses $100p\%$ of agents with the lowest (red) or highest (blue) permanent impact.}
        \label{fig:3}
\end{figure}
     
\subsection{Many agents, choice of portfolio via $\kappa$}

In this subsection, we set $N=100$, fix $\lambda= 10^{-4}$, $\lambda_0 =5\times 10^{-5}$, $\kappa_0=10^{-3}$, generate $\{\kappa_i\}_{i=1}^N$ as independent realizations of a uniform random variable on $(10^{-4},10^{-3})$, and consider the dependence of the broker's local value function $V^{\theta}_P$ on her portfolio of clients $\theta\in\{0,1\}^N$.
As before, we restrict the analysis to the portfolios $\theta$ chosen as the $100p$-percentile of agents with the highest or lowest $\kappa_i$. The local value of the broker for $p\in [0,1]$ is shown in Figure \ref{fig:4}.

Figure \ref{fig:4} shows, once more, that it is not optimal for the broker to take too many clients. It also shows that choosing the agents with large $\kappa$ is better than choosing those with low $\kappa$, which is natural in view of the discussion in Subsection \ref{subse:numerics.2}.

\begin{figure}[H]
        \centering
        \includegraphics[width=0.5\textwidth]{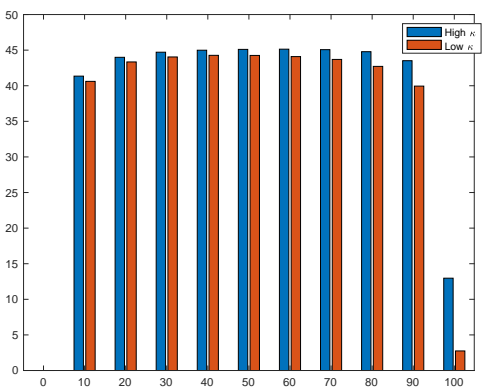}
        \caption{Broker's local value $V^{\theta}_P$ when she chooses $100p\%$ of agents with the lowest (red) or highest (blue) temporary impact.}
        \label{fig:4}
\end{figure}

\subsection{Do agents benefit from the presence of a broker?}
    
In this subsection, we set $N=8$, $\lambda_0=10^{-4}$, $\lambda_i=10^{-4}+\frac{i}{8}(10^{-4}-10^{-3})$, $\kappa_i=10^{-2}+\frac{i}{8}(10^{-1}-10^{-2})$, for $i=1,\ldots,8$, and consider the agents' values as $\kappa_0$ varies over $[0.001,0.01]$. Figure \ref{fig:5} shows the agents' values less their values in the absence of a broker (the latter corresponds to $\theta=0$) -- we refer to them as relative values. Figure \ref{fig:6} shows the optimal portfolio of clients for each $\kappa_0$.

Figure \ref{fig:5} indicates that all agents benefit from the presence of a broker, as their relative values are positive. This is explained by the fact that the broker's price impact is lower than those of the agents. The latter allows some of the agents to reduce their trading costs by becoming broker's clients. As a result, the overall temporary impact on the price is reduced, which benefits the other (independent) agents as well.

Figures \ref{fig:5} and \ref{fig:6} show that the relative values of the agents and the optimal portfolio of clients increase as $\kappa_0$ decreases. This is natural, as small $\kappa_0$ implies larger benefit for each agent who becomes a broker's client and reduces the overall temporary impact on price (which benefits everyone).

\begin{figure}[H]
        \centering
        \includegraphics[width=0.5\textwidth]{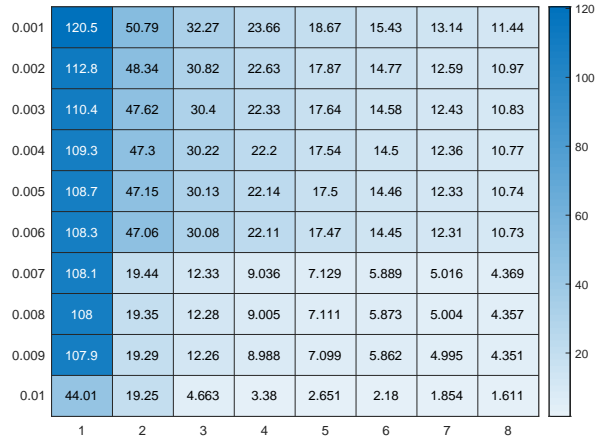}
        \caption{Agents' relative values across different $\kappa_0$ (vertical line) and across agents (horizontal line).}
        \label{fig:5}
\end{figure}    
    
\begin{figure}[H]
        \centering
        \includegraphics[width=0.5\textwidth]{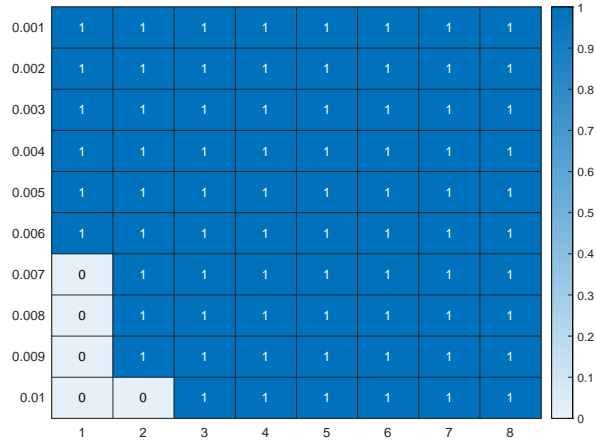}
        \caption{$\theta^*_i$ across different $\kappa_0$ (vertical line) and across $i$ (horizontal line).}
        \label{fig:6}
\end{figure}

\bibliographystyle{plain}
\bibliography{BrokerContract_refs}

\begin{thebibliography}{10}

\bibitem{AlmgrenChriss}
R.~Almgren and N.~Chriss.
\newblock Optimal execution of portfolio transactions.
\newblock {\em Journal of Risk}, 3(2):5--39, 2001.

\bibitem{Frei}
M.~Baldauf, C.~Frei, and J.~Mollner.
\newblock Principal trading arrangements: optimality under temporary and
  permanent price impact.
\newblock Technical report, available at SSRN, Feb 2021.

\bibitem{BasakPavlova}
S.~Basak, A.~Pavlova, and A.~Shapiro.
\newblock Optimal asset allocation and risk shifting in money management.
\newblock {\em The Review of Financial Studies}, 20(5):1583 -- 1621, 2007.

\bibitem{CadenillasCvitanic}
A.~Cadenillas, J.~Cvitani{\'c}, and F.~Zapatero.
\newblock Optimal risk-sharing with effort and project choice.
\newblock {\em Journal of Economic Theory}, 133(1):403 -- 440, 2007.

\bibitem{CasgrainJaimungal}
P.~Casgrain and S.~Jaimungal.
\newblock Mean-field games with differing beliefs for algorithmic trading.
\newblock {\em Mathematical Finance}, 30(3):995 -- 1034, 2020.

\bibitem{CvitanicXing}
J.~Cvitanic and H.~Xing.
\newblock Asset pricing under optimal contracts.
\newblock {\em Journal of Economic Theory}, 173:142 -- 180, 2018.

\bibitem{CvitanicZhang}
J.~Cvitani\'c and J.~Zhang.
\newblock {\em Contract theory in continuous-time models}.
\newblock Springer, 2012.

\bibitem{ElliePossamai}
R.~Ellie and D.~Possama{\"i}.
\newblock Contracting theory with competitive interacting agents.
\newblock {\em SICON}, 57(2):1157 -- 1188, 2019.

\bibitem{HolmstromMilgrom}
B.~Holmstrom and B.~Milgrom.
\newblock Aggregation and linearity in the provision of intertemporal
  incentivese.
\newblock {\em Econometrica}, 55(2):303 -- 328, 1987.

\bibitem{HuangJaimungal}
X.~Huang, S.~Jaimungal, and M.~Nourian.
\newblock Mean-field game strategies for optimal execution.
\newblock {\em Applied Mathematical Finance}, 26(2):153 -- 185, 2019.

\bibitem{NZ}
S.~Nadtochiy and T.~Zariphopoulou.
\newblock Optimal contract for a fund manager with capital injections and
  endogenous trading constraints.
\newblock {\em SIFIN}, 10(3):698--722, 2019.

\bibitem{OuYang}
H.~Ou-Yang.
\newblock Optimal contracts in a continuous-time delegated portfolio management
  problem.
\newblock {\em The Review of Financial Studies}, 16(1):173 -- 208, 2003.

\bibitem{Pham}
H.~Pham.
\newblock {\em Continuous-time Stochastic Control and Optimization with
  Financial Applications}.
\newblock Springer-Verlag Berlin Heidelberg, 2009.

\bibitem{Starks}
L.~T. Starks.
\newblock Performance incentive fees: an agency theoretic approach.
\newblock {\em The Journal of Financial and Quantitative Analysis}, 22(1):17 --
  32, 1987.

\bibitem{Stoughton}
N.~M. Stoughton.
\newblock Moral hazard and the portfolio management problem.
\newblock {\em The Journal of Finance}, 48(5):2009 -- 2028, 1993.

\end{thebibliography}

\end{document}